\documentclass[a4paper, 10pt]{article}

\usepackage[british]{babel}
\usepackage[latin1]{inputenc}
\usepackage[T1]{fontenc}

\usepackage{amsmath,amsthm,amssymb,amstext, amscd}
\usepackage{tikz-cd}
\usepackage{hyperref}
\usepackage{fouriernc}

\usepackage{todonotes}

\theoremstyle{plain}
\newtheorem{Theorem}{Theorem}
\newtheorem{Proposition}[Theorem]{Proposition}

\newtheorem{Lemma}[Theorem]{Lemma}

\theoremstyle{definition}
\newtheorem{Definition}[Theorem]{Definition}
\theoremstyle{remark}

\newcommand{\ie}{\textit{i.e.}, }
\newcommand{\eg}{\textit{e.g.}, }

\newcommand{\R}{\mathbb{R}}
\newcommand{\A}{\mathcal{A}}

\newcommand{\Z}{\mathbb{Z}}
\newcommand{\id}{\operatorname{id}}
\newcommand{\Aut}{\operatorname{Aut}}
\newcommand{\PsAut}{\operatorname{PseudAut}}
\newcommand{\clos}[1]{\overline{#1}}

\newcommand{\Set}[2]{\left\{ #1 \; \mid \; #2 \right\}}
\newcommand{\NP}{\operatorname{NP}}
\newcommand{\N}{\mathbb{N}}
\newcommand{\NN}{\N^\N}

\newcommand{\Si}{\mathbb{S}}
\renewcommand{\O}{\mathcal{O}}
\newcommand{\K}{\mathcal{K}}
\newcommand{\V}{\mathcal{V}}

\title{On the computability of the Fr\'echet distance of surfaces
        in the bit-model of real computation}
\author{Eike Neumann \\Aston University, Birmingham, UK \\ \texttt{neumaef1@aston.ac.uk}}
\makeindex
\date{}

\begin{document}
\maketitle

\abstract
{
	We show that the Fr\'echet distance of two-dimensional parametrised surfaces in a metric space is computable in the bit-model of real computation.
	An analogous result in the real RAM model for piecewise-linear surfaces has recently been obtained by Nayyeri and Xu (2016).
}

\section{Introduction}

	In 1906, Maurice Fr\'echet introduced a natural pseudometric for parametrised curves \cite{FrechetCurve},
	which he generalised in 1924 to parametrised surfaces \cite{FrechetSurface}.
	If $A \colon [0,1]^2 \to X$ and $B \colon [0,1]^2 \to X$ are parametrised surfaces in some metric space,
	then their Fr\'echet distance is given by:

	\begin{equation}\tag{*}\label{Equation: Frechet distance of surfaces}
	\inf_{\varphi, \psi \in \Aut'([0,1]^2)}
		\max_{x \in [0,1]^2}
			d\left( A\left(\varphi(x)\right), B\left(\psi(x)\right) \right),
	\end{equation}

	\noindent
	where $\Aut'([0,1]^2)$ denotes the set of orientation-preserving homeomorphisms of $[0,1]^2$.

	The problem of computing the Fr\'echet distance of curves and surfaces has received considerable attention in Computational Geometry
	(see \cite{AltSurvey} and references therein).
	Alt and Godau showed in 1995 \cite{AltGodau} that the Fr\'echet distance between polygonal curves is polytime computable
	in the usual computational model of computational geometry.
	Later, Godau showed in his PhD-thesis \cite{GodauPhD} that computing the Fr\'echet distance between triangulated surfaces is $\NP$-hard.
	Alt and Buchin \cite{AltBuchin} proved in 2010 that the Fr\'echet distance between triangulated surfaces is upper semicomputable,
	\ie there exists an algorithm which takes an input two triangulated surfaces and returns as output a decreasing sequence of
	rational upper bounds to the Fr\'echet distance, which converges to the true distance.
	Note that no assumption is made on the rate of convergence of this sequence, so that in general one has no information
	on the quality of any such upper bound.
	It was also shown in \cite{AltBuchin} that the so-called weak Fr\'echet distance of surfaces,
	which is obtained by letting $\varphi$ and $\psi$ in \eqref{Equation: Frechet distance of surfaces}
	range over all surjective (and not necessarily injective) reparametrisations, is polytime computable.
	The question remained open whether the Fr\'echet distance of triangulated parametrised surfaces is computable,
	\ie whether there exists an algorithm which takes as input two triangulated surfaces and a number $\varepsilon > 0$
	and produces a rational approximation to the Fr\'echet distance to error $\varepsilon$.
  Recently Nayyeri and Xu \cite{NayyeriXu1} have shown that the Fr\'echet distance of piecewise-linear surfaces
  of genus 0 that are ``locally isometrically''
  immersed in $\mathbb{R}^3$ is computable in the above sense in the
  real RAM model.
  In a follow-up paper \cite{NayyeriXu2}
  they have even shown that the problem of deciding whether
  the Fr\'echet distance between two such surfaces is
  smaller than a given $\delta > 0$ is computable in
  PSPACE.

  In this paper we study the computability of the Fr\'echet distance in the more realistic bit-model of real computation (see \eg \cite{PourElRichards, Weih}) in which real numbers and similar infinite objects are encoded as
  streams of bits rather than being viewed as atomic entities.
  This model takes into account issues such as numerical stability which often
  constitute substantial obstructions to the practical implementability of
  real RAM algorithms.
  We obtain the following result:

	\begin{Theorem}\label{Theorem: main theorem}
	Let $X$ be a computable metric space.
	There exists an algorithm in the bit-model which takes as input
  two parametrised surfaces
	$A \colon [0,1]^2 \to X$,
	$B \colon [0,1]^2 \to X$
	in $X$, and returns as output their Fr\'echet distance.
	\end{Theorem}

  Of course it is impossible to find an algorithm in the bit-model which decides
  for two given surfaces if their distance is smaller than a given $\delta$,
  as in this model equality of real numbers is undecidable.
  Thus Theorem \ref{Theorem: main theorem} is the best one can hope for.
  The techniques used in the proof of Theorem \ref{Theorem: main theorem}
  are quite different from those used by Nayyeri and Xu.
	The main idea behind the proof is to
	compute approximations to the Fr\'echet distance
	by replacing $\Aut'([0,1]^2)$ in \eqref{Equation: Frechet distance of surfaces}
	with a suitable computably compact set,
	using that the minimum of a continuous function over a
	computably compact set is computable.
	The proof can be outlined as follows:
	Section \ref{section: reduction compact} shows that a $2^{-n}$-approximation to the Fr\'echet distance
	can be obtained by letting the infimum in \eqref{Equation: Frechet distance of surfaces} range over the
	set of $L_n$-Lipschitz automorphisms, where $L_n$ is a constant that depends computably on $n$.
	This reduces the problem to the problem of computing the closure $\clos{\Aut'([0,1]^2)}$ of the set of reparametrisations
	as a subset of the space of continuous functions $C([0,1]^2,[0,1]^2)$.
	While this set is relatively easily seen to be lower semicomputable
	(which yields upper semicomputability of the Fr\'echet distance)
	it is more difficult to establish upper semicomputability.
	This amounts to showing that there exists an algorithm which takes as input a map
	$\varphi \colon [0,1]^2 \to [0,1]^2$ and halts if and only if the map is not contained
	in $\clos{\Aut'([0,1]^2)}$.
	While the set of surjective functions is closed, and it is easy to find an algorithm
	which halts if and only if a given function is not surjective
	(which yields computability of the weak Fr\'echet distance),
	falsifying injectivity is considerably more difficult,
	and this is the main part of the proof where some new ideas are needed.
	The main idea is to ``count'' for every $y \in [0,1]^2$ the solutions to the equation $\varphi(x) = y$
	using the Brouwer mapping degree.

	The computability of the Fr\'echet distance in the bit-model was
	also recently studied by Park, Park, Park, Seon, and Ziegler \cite{PPPSZ}.
	They observed that the Fr\'echet distance of continuous curves with values
	in a computable metric space is computable.
	Regarding the question of computational complexity,
	the algorithm we obtain from the proof of Theorem \ref{Theorem: main theorem} makes use of
	multiple unbounded searches and is therefore not even primitive recursive.
	We hence do not obtain any nontrivial upper complexity bounds on the problem beyond establishing
	its computability.
	The problem of characterising the complexity of the Fr\'echet distance in the bit-model is therefore far from settled.
	However, there is some hope that the ideas presented in this paper could be used to design more
	efficient algorithms which yield better upper complexity bounds.

\section{Preliminaries}

	Let us introduce some notation and terminology
	and recall some basic definitions from computable analysis.
	We mainly follow the ideas of Matthias Schr\"oder \cite{SchroederPhD}.
	For a concise introduction to computable analysis see \cite{PaulyRepresented}.
	See \cite{PourElRichards} and \cite{Weih} for classic textbooks on the subject
	which cover some aspects of the theory we require.

	We denote Sierpi\'nski space by $\Si$.
	If $X$ and $Y$ are represented spaces,
	we write $C(X, Y)$ for their exponential in the category of represented spaces.
	We write $\O(X)$ for the open subsets of $X$ identified with $C(X,\Si)$,
	$\A(X)$ for the closed subsets of $X$ identified with their complement as an element of $\O(X)$,
	$\K(X)$ for the (saturated) compacts of $X$ identified in the usual manner with a subspace of $\O(\O(X))$,
	and	$\V(X)$ for the (closed) overts of $X$ identified in the usual manner with a subspace of $\O(\O(X))$.
	A subset of $X$ is called \emph{semi-decidable} if it is a computable point of the space $\O(X)$.
	A closed subset of $X$ is called \emph{lower semicomputable} if it is a computable point of the space $\V(X)$.
	It is called \emph{upper semicomputable} if it is a computable point of the space $\A(X)$.
	A saturated compact subset of $X$ is called \emph{lower semicomputable} if it is a computable point of the space $\V(X)$.
	It is called \emph{upper semicomputable} if it is a computable point of the space $\K(X)$.
	A closed or compact set is called \emph{computable} if it is both lower and upper semicomputable.
	We denote by $\R_<$ the space of \emph{lower reals}, where a real number $x$ is encoded by a sequence $(l_n)_n$ of rational numbers
	which converges from below to $x$.
	We denote by $\R_>$ the space of \emph{upper reals}, where a real number $x$ is encoded by a sequence $(r_n)_n$ of rational numbers
	which converges from above to $x$.
	A real-valued function $F \colon X \to \R$ is called \emph{lower semicomputable}
	if it is computable as a function to the lower reals $\R_<$ and \emph{upper semicomputable} if
	it is computable as a function to the upper reals $\R_>$.
	If $f \colon \subseteq X \to Y$ is a partial function, we will
	often say that $f(x)$ is \emph{uniformly computable} in $x$
	to express that $f$ is a computable function.
	Throughout this paper we endow $\R^n$ with the maximum norm
	\[|x| = \max\Set{|x_i|}{i = 1, \dots, n}. \]
	It will be convenient to write $D^n$ for the unit cube $[0,1]^n$
	and $S^{n - 1}$ for its boundary $\partial [0,1]^n$.
	If $A$ is a subset of a metric space $M$, let
	\[d(x,A) = \inf \Set{d(x,y)}{y \in A} \]
	denote the \emph{distance function} of $A$.
	If $\varepsilon > 0$ is a real number, we call the set
	\[A^{\clos{\varepsilon}} = \Set{x \in M}{d(x,A) \leq \varepsilon} \]
	the \emph{closed $\varepsilon$-thickening} of $A$.
	Analogously, the set
	\[A^{\varepsilon} = \Set{x \in M}{d(x,A) < \varepsilon} \]
	is called the \emph{open $\varepsilon$-thickening} of $A$.

	Our proof of Theorem \ref{Theorem: main theorem} is based on the following
	simple observation:
	\begin{Proposition}\label{Proposition: Computability of minimum}
	\hfill
	\begin{enumerate}
	\item
		The infimum of an upper semicomputable function over a lower semicomputable closed set
		is uniformly upper semicomputable.
		More formally, for every represented space $X$, the function
		\[\inf \colon \V(X) \times C(X,\R_>) \to \R_>, \; (A,f) \mapsto \inf\Set{f(x)}{x \in A} \]
		is computable.
	\item
		The minimum of a lower semicomputable function over an upper semicomputable compact set
		is uniformly lower semicomputable.
		More formally, for every represented space $X$, the function
		\[\min \colon \K(X) \times C(X, \R_<) \to \R_<, \; (K,f) \mapsto \min\Set{f(x)}{x \in K} \]
		is computable.
	\end{enumerate}
	\end{Proposition}

\section{The Brouwer mapping degree}

	The main topological tool for proving Theorem \ref{Theorem: main theorem} will be the Brouwer mapping degree.
	We will summarise here the main facts we need in the sequel.
	Very readable constructions of the degree are given in \cite{MilnorDiffTopo} and \cite{Teschl}.

	\begin{Theorem}
		There exists a unique function
		\[\deg \colon \subseteq C(\R^n,\R^n) \times \mathcal{O}(\R^n) \times \R^n \to \Z\]
		with domain
		\[\operatorname{dom}(\deg) = \Set{(f,U,y)}{U \text{ bounded, } y \notin f(\partial U)} \]
		satisfying the following properties:
		\begin{enumerate}
		\item
			\textsc{Translation invariance:}
				$\deg(f,U,y) = \deg(f - y, U, 0)$.
		\item
			\textsc{Normalisation:}
			$\deg(\id, U, y) = 1$ for all $ y \in U$.
		\item
			\textsc{Additivity:}
			If $U_1$ and $U_2$ are open disjoint subsets of $U$ with
			$y \notin f\left(\clos{U} \setminus \left(U_1 \cup 	U_2 \right)\right)$
			then
			$\deg(f,U,y) = \deg(f,U_1,y) + \deg(f,U_2,y)$.
		\item
			\textsc{Homotopy invariance:}
			If $H(t, x)$ is a homotopy from $f$ to $g$ with
			$y \notin H(t,\partial U)$ for all $t \in [0,1]$
			then
			$\deg(f,U,y) = \deg(g,U,y)$.
		\end{enumerate}
	\end{Theorem}

	\begin{Proposition}
	If $\deg(f,U,y)$ is well-defined and non-zero, then the equation $f(x) = y$ has a solution in $U$.
	\end{Proposition}

	It can be shown that the degree is computable when the open sets are correctly topologised.
	Let $\mathcal{U}	(\R^n)$ denote the space of open subsets of $\R^n$ which is obtained by
	identifying an open set $U$ with its two-sided distance function:
	\[
	d_{\text{two-sided}}(\cdot , U) \colon X \to \R, \; x \mapsto
					\left\{
					\begin{array}{lr}
						d(x, \partial U) & \text{if }x \notin U,\\
						-d(x, \partial U) & \text{if }x \in U.
					\end{array}
					\right.
	\]
	Note that the underlying representation is much stronger than the standard representation of open sets.

	\begin{Theorem}\label{Theorem: Computability of the mapping degree}
	The partial map
	\[\deg \colon \subseteq C(\R^n) \times \mathcal{U}(\R^n) \times \R^n, \; (f,U,y) \mapsto \deg(f,U,y) \]
	is computable with semi-decidable domain.
	\end{Theorem}
	\begin{proof}[Proof Sketch]
	The degree $\deg(f,U,y)$ is defined so long as $y \notin f(\partial U)$, and this is uniformly semi-decidable
	for continuous $f$ and $U \in \mathcal{U}$.
	To compute $\deg(f,U,y)$, compute a sufficiently good twice differentiable approximation $\tilde{f}$ to $f$
	and a sufficiently good approximation $\tilde{y}$ to $y$, which is a regular value of $\tilde{f}$.
	It suffices to choose $\tilde{y}$ with
	$\left|y - \tilde{y} \right| < d\left(y, f(\partial U)\right)$
	and $\tilde{f}$ with
	$\left|f - \tilde{f} \right| < d\left(y, f(\partial U)\right)$.
	The fact that $\tilde{y}$ can be chosen to be a regular value follows from Sard's theorem.
	Then $\deg(f,U,y)$ can be computed using the determinant formula:
	\[\deg(f,U,y) = \deg(\tilde{f}, U, \tilde{y}) =
		\sum_{x \in \tilde{f}^{-1}(y)} \operatorname{sgn} \left(\det\left(D\tilde{f}(x)\right)\right). \]
	For more details refer to the construction of the mapping degree in \cite[Chapter 16]{Teschl}.
	\end{proof}

	A similar result was proved by Miller \cite{MillerPhD} for the fixed point index on rational cubical complexes
	(see also \cite{BrouwerConnectedChoice1, BrouwerConnectedChoice2, BrouwerConnectedChoiceArxiv}).
	His proof uses computational homology rather than the determinant formula.
	An analogous result based on computational homology is stated in Collins \cite{CollinsZeroSet}.

\section{Reduction to a compact search problem}\label{section: reduction compact}

	For a map $f \colon X \to Y$, let $\Gamma_f$ denote its graph.
	Define a new distance function $d_\Gamma$ on $C(D^2,D^2)$ by
	\[d_\Gamma(f,g) = d_H(\Gamma_f, \Gamma_g) \]
	where $d_H$ is the Hausdorff distance on the metric space $D^2 \times D^2$ with the product metric
	\[d\left((x_0,y_0), (x_1,y_1)\right) = \max \left\{d(x_0,x_1), d(y_0, y_1)\right\}.\]
	We will call this the \emph{graph distance} on $C(D^2, D^2)$.
	With respect to $d_\Gamma$ the metric space $C(D^2,D^2)$ is totally bounded but incomplete.
	The total boundedness is what will allow us to reduce the problem of computing the Fr\'echet
	distance to the problem of computing a minimum over a compact set.
	A similar idea is used in \cite{PPPSZ} to compute the Fr\'echet distance of curves.

	Recall that if $f \colon X \to Y$ is a uniformly continuous map between metric spaces, then a
	\emph{modulus of (uniform) continuity} for $f$ is a function $\omega \colon \N \to \N$ such that
	for all $x, y \in X$ we have the implication:
	\[d(x, y) \leq 2^{-\omega(n)} \rightarrow d\left(f(x), f(y)\right) < 2^{-n}. \]
	The graph distance is useful, as good approximations of the reparametrisations
	with respect to the graph distance yield good approximations of the Fr\'echet distance.

	\begin{Lemma}\label{Lemma: Graph distance}
	For a function $\varphi \colon D^2 \to D^2$, let
	\[F_{A,B}(\varphi) = \max_{x \in D^2} d\left( A(x), B(\varphi(x)) \right). \]
	Let $\mu_A$ and $\mu_B$ be moduli of continuity of $A$ and $B$ respectively.
	Then we have the implication
	\[d_{\Gamma}(\varphi,\psi) \leq 2^{-\mu_A(n + 1) - \mu_B(n + 1)}
	 \rightarrow
	 d\left(F_{A,B}(\varphi), F_{A,B}(\psi)\right) \leq 2^{-n}.\]
	\end{Lemma}

	\begin{Lemma}\label{Lemma: Lipschitz constant increase of radial extension}
	Let $f \colon \partial [0,1]^n \to \R^n$ be a map.
	Let $\tilde{f} \colon [0,1]^n \to \R^n$ denote its \emph{radial extension}
	\[\tilde{f}(x) =
		\begin{cases}
		0 														&\text{if } x = 0\text{,} \\
		|x| \cdot f\left(\tfrac{x}{|x|}\right) &\text{otherwise.}
		\end{cases} \]
	If $L$ is a Lipschitz constant for $f$, then
	$L + |f|$ is a Lipschitz constant for $\tilde{f}$.
	\end{Lemma}
	\begin{proof}
	Let $x,y \in [0,1]^n$ with $|y| \leq |x|$.
	Let $r = |x|$ and $s = |y|$.
	If $s = 0$ then
	\[\left|\tilde{f}(x) - \tilde{f}(y)\right| = r \cdot \left| f(\tfrac{x}{r})\right| \leq |x - y| \cdot |f|. \]
	If $s > 0$, we calculate:
	\begin{align*}
	\left| \tilde{f}(x) - \tilde{f}(y) \right|
		&=	 	\left| r\cdot f(\tfrac{x}{r}) - s\cdot f(\tfrac{y}{s}) \right|  \\
		&\leq		\left| s\cdot f(\tfrac{x}{r}) - s\cdot f(\tfrac{y}{s}) \right|
			    + \left| r\cdot f(\tfrac{x}{r}) - s\cdot f(\tfrac{x}{r}) \right|  \\
		&\leq	 	s\cdot L \cdot \left| \tfrac{x}{r} - \tfrac{y}{s} \right|
			+ |f| \cdot |r - s|		\\
		&\leq L \cdot |\tfrac{s}{r} x - y| + |f| \cdot |x - y|. \\
		&\leq L \cdot |x - y| + |f| \cdot |x - y|.
	\end{align*}
	For the last line, note that the point $\tfrac{s}{r} x$ is the projection of $x$
	onto the set $\Set{z \in \R^n}{|z| = s}$.
	\end{proof}

	The following Lemma is the main result of this section.
	It shows that the space of automorphisms of $D^2$ is totally
	bounded with respect to the graph distance.

	\begin{Lemma}\label{Lemma: Lipschitz approximation}
	Let $f \colon D^2 \to D^2$ be an automorphism.
	For all $n \in \N$
	there exists an automorphism
	$\tilde{f} \colon D^2 \to D^2$
	with $d_{\Gamma}(f, \tilde{f}) < 2^{-n}$
	such that $\tilde{f}$ has Lipschitz constant
	\[4^n \times 4^{4^n} \times \left(3 \times 4^n + 3\right) + 1.\]
	\end{Lemma}
	\begin{proof}
	Subdivide $D^2$ into a uniform square grid of mesh width $2^{-n}$.
	We will replace $f$ on the edges of this grid by a map with small Lipschitz constant
	and extend this map radially to obtain the desired approximation.
	An edge of the grid is simply an edge of one of the squares of the subdivision.
	We view these edges as subsets of $D^2$.
	Hence if two squares meet at an edge, this edge is counted as one edge and not as two.
	An edge is called an \emph{interior edge} if it is not completely contained in the boundary of $D^2$.
	Otherwise it is called a \emph{boundary edge}.

	Consider an interior edge $e$ of this grid.
	The map $f$ sends $e$ to a simple curve $C$.
	We can isometrically identify $e$ with the interval $[0, 2^{-n}]$ and think
	of the curve $C$ as being parametrised over this interval by a continuous function
	$\gamma \colon [0, 2^{-n}] \to D^2$.

	Up to slightly perturbing $f$ we can assume that all
	curves $C$ of this form
	intersect the mesh in a non-degenerate manner
	in the sense that the intersection of $C$ with the edges of the mesh
	is zero-dimensional, \ie $\gamma$ never maps an interval into an edge,
	and that $C$ does not intersect any vertices of the mesh.

	By the assumption that $C$ intersects the mesh in a non-degenerate manner
	there exists a unique square $S_0$ such that the initial segment $\gamma([0, \delta))$
	of the curve is completely contained in $S_0$ for some $\delta > 0$.

	For a square $S$
	we say that $\gamma$ is \emph{staying at} $S$ in the interval $[a, b]$ if
	$\gamma(t) \notin S$ for $t < a$,
	$\gamma(a + \delta) \in S$, for all sufficiently small $\delta > 0$,
	$\gamma(b) \in S$
	and $\gamma(t) \notin S$ for all $t > b$.

	Note that the intervals $[a,b]$ in which $\gamma$ is staying at some square $S$
	form a forest $\mathcal{F}_\gamma$ (\ie a finite union of trees)
	with respect to the usual inclusion order.
	The interval $[0, 2^{-n}]$ decomposes into the maximal intervals of this order.

	If $\gamma$ is staying at $S$ in the interval $[a,b]$, we call a restriction
	$\gamma\left|_{[c,d]}\right.$ with $[c,d] \subseteq [a,b]$
	an \emph{arm} of $\gamma$ in $[a,b]$
	if $\gamma(c) \in S$ and $\gamma(d) \in S$,
	but
	$\gamma(t) \notin S$ for all $c < t < d$.
	We call the interval $[c,d]$ the \emph{domain} of the arm $A$.
	The arms of $\gamma$ in $[a,b]$ can be linearly ordered by comparing the left
	endpoints of their domains with respect to the usual order on $\R$.
	If $S'$ is another square and $A$ is an arm, we say that $A$ \emph{reaches} $S'$
	if it intersects $S'$.
	We say that an arm $A$ of $\gamma$ in $[a,b]$ \emph{reaches new squares} if there exists a square $S'$
	such that $A$ reaches $S'$ but no arm $A'$ of $\gamma$ in $[a,b]$ with $A' \leq A$ reaches $S'$.

	We replace $\gamma$ with a simpler curve that only has arms which reach new squares.
	Let $I$ be a maximal interval in the forest $\mathcal{F}_\gamma$
	and let $S$ be the square at which $\gamma$ is staying in $I$.
	Let $A$ be an arm of $\gamma$ in $I$ which does not reach any new squares.
	Let $[c,d]$ be the domain of $A$.
	Then there exist $c', d' \in I$ with
	$c' < c < d < d'$
	such that $A$ is the only arm of $\gamma$ in the interval
	$[c',d']$ and $\gamma(c')$ and $\gamma(d')$ are interior points of $S$.
	Hence, the arm $A$ can be pruned away by replacing $\gamma$ on $[c',d']$
	with the linear interpolation in $\gamma(c')$ and $\gamma(d')$.
	Use this method to prune away all arms in $I$ which do not reach new squares.
	Note that we can do this in such a way that the new curve does not intersect itself.
	Do the same for all maximal intervals in $\mathcal{F}_\gamma$ and
	apply this procedure recursively to all intervals on the next level,
	so that we eventually obtain a
	new curve $\zeta$ which does not	have any arms that do not reach new squares.

	Replace $\zeta$ with a suitable linear approximation $\eta$ which is constructed as follows:
	consider the intersection of $\zeta$ with the edges of our grid.
	Since $\zeta$ intersects the grid in a nondegenerate manner
	this is a finite set of points.
	Order them according to the order in which they are visited by $\zeta$.
	If the line segment between two consecutive points $\zeta(t_0)$ and $\zeta(t_1)$
	is contained in an edge of the grid,	add an additional point $\zeta(s)$ with
	$s \in (t_0, t_1)$ to the set.
	This point is necessarily contained in the interior of some square.
	Consider the polygonal chain $P$ which interpolates these points in the given order.
	This chain could have some self-intersections.
	Note however that if we have two segments in the chain whose endpoints lie
	on the edges of the grid, then the two segments intersect if and only if
	the original curve $\zeta$ has a self-intersection.
	Hence, the only self-intersections of $P$ can happen between segments where at least
	one endpoint is an interior point of a square.
	In this case, we can move this endpoint closer to the boundary of the square
	to resolve the intersection.
	Doing this finitely many times yields a simple polygonal chain which intersects
	the same squares as $\zeta$.
	Let $\eta \colon [0, 2^{-n}] \to D^2$
	be the parametrisation of this chain where each segment is traversed at the same speed.

	Let us now estimate the Lipschitz constant of the curve $\eta$.
	As there are $4^n$ squares in total,
	the forest $\mathcal{F}_{\eta}$ contains at most $4^n$ trees.
	By the same argument each tree has height at most $4^n$.
	Since $\eta$ intersects the mesh in a non-degenerate manner,
	each node in the tree has at most four children.
	Hence each tree has at most $4^{4^n}$ elements,
	and thus the forest has at most $4^n \times 4^{4^n}$ elements in total.
	By construction, the number of vertices used in the linear interpolation of $\zeta$
	in each element of the forest is bounded by
	\[3 \times \left(\text{ the number of arms}\right) + 3 .\]
	As every arm has to reach at least one new square, there are at most $4^n$ arms,
	so that every element of the forest contributes at most
	$3 \times 4^n + 3$ vertices.
	In total there are at most
	\[N = 4^n \times 4^{4^n} \times \left(3 \times 4^n + 3\right)\]
	vertices, and just as many line segments.
	Thus, the interval $[0, 2^{-n}]$ is divided into $N$ segments of length
	$2^{-n} / N$.
	Each line segment has diameter at most $2^{-n}$.
	Hence, the Lipschitz constant of $\eta$ is bounded by $N$.

	If $e$ is a boundary edge, then $f(e)$ is contained in the boundary of $[0,1]$.
	Hence, we can construct a piecewise linear function with the same image, which
	uses at most $5$ pieces.
	Its Lipschitz constant is therefore bounded by $5\cdot 2^n < N$.

	We have constructed a piecewise linear curve $\eta$ for every edge $e$ of the mesh.
	Similarly as in the construction of each individual curve, we can make sure
	that these curves do not intersect by potentially moving certain vertices
	in the interiors of the squares of the grid closer to the boundary
	(which does not change the estimate of the Lipschitz constant).
	Hence, these curves define a bijective map $\tilde{f}$ on the $1$-skeleton of the mesh.
	By Lemma \ref{Lemma: Lipschitz constant increase of radial extension}
	this map extends to a bijection with Lipschitz constant $N + 1$ via radial extension
	to the interiors of the squares.
	By construction, the curve $\eta$ intersects a square if and only if $\gamma$ does so.
	Hence, the image of every square $S$ under $\tilde{f}$ is $2^{-n}$-close in the Hausdorff distance to the
	image of $f$.
	Since every square has diameter at most $2^{-n}$, the graphs of $f$ and $\tilde{f}$ are
	$2^{-n}$ close in the Hausdorff distance.
	Hence, $\tilde{f}$ has all the desired properties.
	\end{proof}

	Finally, we observe that the compact space of $L$-Lipschitz functions is
	uniformly upper semicomputably compact in $L$ as a subset of $C(D^2,D^2)$.
	This is a special case of the constructive version of the Arzel\`a-Ascoli theorem, which was
	proved in \cite{BishopBridges}.

	\begin{Theorem}\label{Theorem: computable Arzela-Ascoli}
	The map
	\begin{align*}
	&\operatorname{mod} \colon \NN \to \K\left(C\left(D^2, D^2\right)\right), \\
		&\omega \mapsto \Set{f \colon D^2 \to D^2}{\omega \text{ is a modulus of continuity for } f}
	\end{align*}
	is computable.
	In particular, the map
	\begin{align*}
	&\operatorname{lip} \colon \R \to \K\left(C\left(D^2, D^2\right)\right), \\
	&L \mapsto \Set{f \colon D^2 \to D^2}{\forall x,y \in D^2. \left(|f(x) - f(y)| \leq L \cdot |x - y| \right)}
	\end{align*}
	is computable.
	\end{Theorem}

\section{Computability of automorphisms}

	The goal of this section is to establish the following result:

	\begin{Theorem}\label{Theorem: automorphisms located}
		The closure $\clos{\Aut'(D^2)}$ of the set of orientation-preserving
		automorphisms of the unit square $D^2$	is computable as a closed subset of
		the space $C(D^2,D^2)$ of continuous self-maps of $D^2$.
	\end{Theorem}
	\begin{proof}
		This follows from Theorems \ref{Theorem: pseudo-aut located} and \ref{Theorem: pseudo-aut = clos(aut)} below.
	\end{proof}

	Let $f \colon S^1 \to S^1$ be a self-map of the unit circle.
	Then $f$ is called an \emph{orientation-preserving pseudo-automorphism}
	(or just pseudo-\-au\-to\-mor\-phism for short)
	if $f$ lifts to a surjective monotonically increasing map
	$\tilde{f} \colon [0,1] \to [0,1]$
	with respect to suitable orientation-preserving
	bijective parametrisations of its domain and codomain.
	Note that being a pseudo-\-au\-to\-mor\-phi\-sm is computably falsifiable,
	\ie the pseudo-\-au\-to\-mor\-phi\-sms are an upper
	semicomputable closed subset of $C(S^1,S^1)$.
	It is easy to see that they are also lower semicomputable.

	\begin{Definition}\label{Definition: Pseudo-automorphism}
	For $f \colon D^2 \to D^2$, let $\tilde{f} \colon \R^2 \to \R^2$ denote the radial extension
	\[\tilde{f}(r \cdot x) = r\cdot f(x), \; \text{ where } x \in \partial D^2, \; r \in [1, +\infty). \]
	Call $f \colon D^2 \to D^2$ an orientation preserving pseudo-automorphism
	(or just pseudo-automorphism)
	if it satisfies the following conditions:
	\begin{enumerate}
	\item
		$f(\partial D^2) \subseteq \partial D^2$.
	\item
		$f|_{\partial D^2} \colon S^1 \to S^1$ is a pseudo-automorphism of the unit circle.
	\item
		$f$ is surjective.
	\item
		If $U_1, \dots, U_n$ are disjoint open subsets of $\R^2$ and
		$y \in \tilde{f}(\clos{\bigcup U_i}) \setminus \tilde{f}(\partial \bigcup U_i)$
		then
		$\sum_{i = 1}^n \deg(\tilde{f},U_i,y) = 1$.
	\end{enumerate}
	\end{Definition}

	Let $\PsAut^2$ denote the set of all (orientation preserving) pseudo-\-au\-to\-mor\-phisms of $D^2$.

	\begin{Theorem}\label{Theorem: pseudo-aut located}
		The set $\PsAut^2$ is computable as a closed subset of the space $C(D^2,D^2)$.
	\end{Theorem}
	\begin{proof}
		This essentially follows from the definition.
		On the one hand, all the conditions in Definition \ref{Definition: Pseudo-automorphism} define upper semicomputable closed sets.
		Thus their intersection is again upper semicomputable.
		On the other hand,
		piecewise linear automorphisms
		which are given by matrices with rational entries
		on a rational polygonal subdivision of
		$D^2$ are dense in the set of all automorphisms,
		which are in turn dense in the set of all pseudo-automorphisms by
		Theorem \ref{Theorem: pseudo-aut = clos(aut)} below.
		As we can semi-decide for a given piecewise linear map which is specified
		by rational data if it is bijective, the set $\PsAut^2$ admits a computably
		enumerable dense sequence and	therefore is lower semicomputable.
	\end{proof}

	\begin{Theorem}\label{Theorem: pseudo-aut = clos(aut)}
		We have $\PsAut^2 = \clos{\Aut'(D^2)}$.
	\end{Theorem}
	\begin{proof}
		Clearly every automorphism is a pseudo-automorphism.
		The proof of the converse takes up the rest of this section.
		It follows from Lemma \ref{Lemma: approximating pseudo-auto} below.
	\end{proof}

	We now prove the remaining direction of Theorem \ref{Theorem: pseudo-aut = clos(aut)}.
	Given a pseudo-automorphism $f \colon D^2 \to D^2$ and a number $\varepsilon > 0$
	we construct an automorphism $f_{\varepsilon}$ which is $\varepsilon$-close
	to it.
	The key observation is the following:

	\begin{Lemma}\label{Lemma: set of solutions connected and simply connected}
		Let $f \colon D^2 \to D^2$ be a pseudo-automorphism.
		Then the preimage under $f$ of every connected set is connected
		and the preimage under $f$ of every simply connected set is simply connected.
	\end{Lemma}
	\begin{proof}
		Let $C$ be connected.
		Let $A_1 \cap A_2 \supseteq C$ be a partition of $f^{-1}(C)$ into
		non-empty closed subsets $A_1$ and $A_2$.
		Then $f(A_1) \cup f(A_2)$ is a partition of $C$ into non-empty closed subsets.
		Since $C$ is connected, there exists $y \in f(A_1) \cap f(A_2)$.
		If $A_1$ and $A_2$ are disjoint, they can be separated by open neighbourhoods
		$U_1 \supseteq A_1$ and $U_2 \supseteq A_2$ in $\R^2$.
		Recall that $\tilde{f}$ denotes the radial extension of $f$ to $\R^2$.
		Since $f$ is a pseudo-automorphism we have
		$\deg(\tilde{f}, U_i, y) = 1$ for $i = 1,2$
		and $\deg(\tilde{f}, U_1, y) + \deg(\tilde{f},U_2,y) = 1$.
		Contradiction!
		It follows that $A_1 \cap A_2 \neq \emptyset$, and so $f^{-1}(C)$ is
		connected.

		Now, assume that $C$ is simply connected.
		Let $L$ be a loop in $f^{-1}(C)$.
		If $L$ cannot be deformed into a point, then there exists
		$x$ in the region $U$ which is bounded by $L$ with $f(x) \notin C$.
		In particular the degree $\deg(f,U,f(x))$ is well defined.
		As $f \circ L$ is a closed curve in $C$ and $C$ is simply connected it
		follows from $f(x) \notin C$ that $f(x)$ is not contained in any
		region that is bounded by $f \circ L$.
		This implies $\deg(f, U, f(x)) = 0$.
		On the other hand, since $f$ is a pseudo-automorphism we have
		$\deg(f,U,f(x)) = 1$.
		Contradiction!
		It follows that $f^{-1}(C)$ is simply connected.
	\end{proof}

	We now construct $f_{\varepsilon}$ based on
	Lemma \ref{Lemma: set of solutions connected and simply connected}.
	We divide the codomain $D^2$ into a uniform square mesh $\mathcal{M}$ of
	mesh width $h < \varepsilon/2$.
	Reintroducing some notation from the proof of Lemma \ref{Lemma: Lipschitz approximation},
	call a vertex of the mesh an \emph{interior vertex} if it is contained in the
	interior of $D^2$.
	Call it a \emph{boundary vertex} otherwise.
	Call an edge an \emph{interior edge}, if at least one of its endpoints is
	an interior vertex.
	Call it a \emph{boundary edge} otherwise.
	Assign to each interior vertex $v$ an arbitrarily chosen point
	$x(v) \in f^{-1}(\{v\})$.
	Assign to each boundary vertex $v$ an arbitrarily chosen point
	$x(v) \in f^{-1}(\{v\}) \cap \partial D^2$.
	Such a point exists since $f|_{\partial D^2}$ is
	a surjective map onto the boundary.
	For each boundary edge $e$ with endpoints $v_0$ and $v_1$,
	let $A(e)$ denote the arc in the boundary	that joins the
	two points $x(v_0)$ and $x(v_1)$.
	As $f|_{\partial D^2}$ is pseudo-automorphism of circles,
	different boundary edges $e_0$ and $e_1$ are associated with arcs
	$A(e_0)$ and $A(e_1)$ that do not intersect except in their endpoints.
	Now consider an interior edge $e$ with endpoints $v_0$ and $v_1$.
	Let $U = e^{h/4}$ be the open $h/4$-thickening of $e$.
	By Lemma \ref{Lemma: set of solutions connected and simply connected}, the
	set $f^{-1}(U)$ is connected.
	In particular there exists an arc $A(e)$ in $f^{-1}(U)$
	which connects $x(v_0)$ and $x(v_1)$.
	We can choose these arcs such that arcs associated with different edges
	do not intersect except in their endpoints.
	Then for each square $S \in \mathcal{M}$ of the mesh we have four arcs,
	corresponding to the four edges of the square,
	which form a simple closed curve $C(S) \subseteq D^2$.
	This curve bounds an open region $\Omega(S) \subseteq D^2$.
	Now, choose for each edge $e$ with endpoints $v_0$ and $v_1$
	a bijective map
	$\varphi_e \colon A(e) \to e$
	which sends $x(v_0)$ to $v_0$	and $x(v_1)$ to $v_1$.
	For each square $S \in \mathcal{M}$ this yields a bijective map
	$\varphi_S \colon \clos{\Omega(S)} \to S$ by radial extension.
	Now assign to each $x \in D^2$ which is contained in some region
	$\clos{\Omega(S)}$ the value $\varphi_S(x)$.
	This defines a partial map
	\[ f_{\varepsilon} \colon \subseteq D^2 \to D^2. \]
	Our next goal is to show that this map is well-defined,
	bijective, and $\varepsilon$-close to $f$.

	\begin{Lemma}\label{Lemma: preimage preserves topology}
		\hfill
	  \begin{enumerate}
			\item
				If $S_0$ and $S_1$ are adjacent squares which intersect only in an
				edge $e$, then the corresponding regions $\clos{\Omega(S_0)}$
				and $\clos{\Omega(S_1)}$ intersect only in the arc $A(e)$.
				Furthermore, the arc $A(e)$ without its endpoints is contained
				in the interior of the
				union of  $\clos{\Omega(S_0)}$ and  $\clos{\Omega(S_1)}$.
			\item
				If $S_0$ and $S_1$ are adjacent squares which intersect only
				in a vertex $v$, then the corresponding regions $\clos{\Omega(S_0)}$
				and $\clos{\Omega(S_1)}$ intersect only in the point $x(v)$.
			\item
				If $S_0$ and $S_1$ are non-adjacent squares, then the corresponding regions
				$\clos{\Omega(S_0)}$ and $\clos{\Omega(S_1)}$ are disjoint.
		\end{enumerate}
	\end{Lemma}
	\begin{proof}
		\hfill
		\begin{enumerate}
			\item
				The curves $C(S_0)$ and $C(S_1)$ intersect only in $A(e)$ by
				construction.
				Hence, if $\clos{\Omega(S_0)}$ and $\clos{\Omega(S_1)}$ intersect in
				further points, one set must contain the other.
				Let $c_0$ denote the centre of $S_0$.
				Then
				$\clos{\Omega(S_0)}$ contains $f^{-1}(c_0)$
				but
				$\clos{\Omega(S_1)}$ is disjoint from $f^{-1}(c_0)$.
				Hence, $\clos{\Omega(S_1)}$ cannot contain $\clos{\Omega(S_0)}$.
				By symmetry, $\clos{\Omega(S_0)}$ cannot contain $\clos{\Omega(S_1)}$.
				It follows that $\clos{\Omega(S_0)}$ and $\clos{\Omega(S_1)}$ intersect
				each other only in $A(e)$.
				The arc $A(e)$ is then contained in the bounded region of the curve
				whose trace is the union of the remaining arcs of $C(S_0)$ and $C(S_1)$.
				It follows that $A(e)$ without its endpoints
				is contained in the interior of the
				union of  $\clos{\Omega(S_0)}$ and  $\clos{\Omega(S_1)}$.
			\item
				Follows from an analogous argument.
			\item
				By construction,
				$\clos{\Omega(S_0)}$ and $\clos{\Omega(S_1)}$ are
				contained in the preimages of the $h/4$-thickenings of $S_0$ and $S_1$
				respectively.
				If $S_0$ and $S_1$ are non-adjacent then these sets are disjoint,
				and then so are $\clos{\Omega(S_0)}$ and $\clos{\Omega(S_1)}$.
		\end{enumerate}
	\end{proof}

	\begin{Lemma}\label{Lemma: approximating pseudo-auto}
		The map $f_{\varepsilon}$ is well-defined, total, bijective, and
		$\varepsilon$-close to $f$.
	\end{Lemma}
	\begin{proof}
		By Lemma \ref{Lemma: preimage preserves topology},
		the map is well-defined, as every point in its domain is assigned
		at most one value.
		The domain of $f_{\varepsilon}$ is a finite union of closed sets, and hence
		closed.
		If $x$ is a point in the domain, it is either contained in an open region
		$\Omega(S)$ or on an arc $A(e)$.
		In the latter case it is either contained in the boundary $\partial D^2$
		or it is contained in a common arc of two regions $\clos{\Omega(S_0)}$
		and $\clos{\Omega(S_1)}$.
		If it is contained in an arc, it is either an endpoint of the arc or an
		``interior point'', \ie not an endpoint.
		If it is an ``interior point'', by Lemma \ref{Lemma: preimage preserves topology}.1
		it is contained in the interior of the union of the two regions, and hence
		in particular in the interior of the domain of $f_{\varepsilon}$.
		If it is an endpoint, then by a similar argument as in
		Lemma \ref{Lemma: preimage preserves topology}.1
		it is contained in the interior of the union of the four regions
		which meet at this point.
		Hence, every point of the domain of $f_{\varepsilon}$ is an interior point.
		It follows that the domain of $f_{\varepsilon}$ is open and closed.
		Hence $f_{\varepsilon}$ is total.
		By construction $f_{\varepsilon}$ is bijective.
		Finally, if $x \in D^2$ is mapped by $f$ to a square $S$, then
		$f_{\varepsilon}(x)$ is contained in an adjacent square.
		Hence, $f_{\varepsilon}$ is $\varepsilon$-close to $f$.
	\end{proof}

\section{Computability of the Fr\'echet distance}

	Putting it all together we can prove Theorem \ref{Theorem: main theorem}.
	Let us restate the result more formally:

	\begin{Theorem}
		Let $X$ be a computable metric space.
		Then the function
		\begin{align*}
		d_{\text{Fr\'echet}} \colon &C(D^2, X) \times C(D^2, X) \to \R, \\
		&(A,B) \mapsto
			\inf_{\varphi, \psi \in \Aut'(D^2)}
				\max_{x \in D^2} d\left( A\left(\varphi(x)\right), B\left(\psi(x)\right) \right)
		\end{align*}
		is computable.
	\end{Theorem}
	\begin{proof}
		It follows from Theorem \ref{Theorem: automorphisms located} and
		Proposition \ref{Proposition: Computability of minimum}.1
		that the Fr\'echet distance is uniformly upper semicomputable in $A$ and $B$.
		It remains to show that the Fr\'echet distance is uniformly lower semicomputable in $A$ and $B$.
		For a number $L \in \R$, let
		\[	\clos{\Aut'_L(D^2)} =
		\clos{\Aut'(D^2)} \cap \Set{f \colon D^2 \to D^2}{\forall x,y \in D^2.\left( \left| f(x) - f(y) \right| \leq L \cdot \left| x - y \right| \right)}\]
		denote the closure of the set of $L$-Lipschitz orientation-preserving automorphisms of $D^2$.
		Let $n \in \N$.
		By Lemma \ref{Lemma: Lipschitz approximation} and Lemma \ref{Lemma: Graph distance}, the number
		\[d_n = \inf_{\varphi, \psi \in \clos{\Aut'_{L_n}(D^2)}} \max_{x \in D^2} d\left( A(\varphi(x)), B(\psi(x)) \right),\]
		where
		\[L_n = 4^{\alpha(n)} \times 4^{4^{\alpha(n)}} \times \left(3 \times 4^{\alpha(n)} + 3\right) + 1\]
		and
		\[\alpha(n) = \mu_A(n + 1) + \mu_B(n + 1)\]
		is $2^{-n}$-close to the Fr\'echet distance $d_{\text{Fr\'echet}}(A,B)$.
		By Theorem \ref{Theorem: automorphisms located} and Theorem \ref{Theorem: computable Arzela-Ascoli},
		the set $\clos{\Aut'_{L_n}(D^2)}$ is uniformly computably compact in $n$.
		Hence, by Proposition \ref{Proposition: Computability of minimum}.2
		the numbers $d_n - 2^{-n}$ are uniformly lower semicomputable in $A$, $B$, and $n$.
		As these numbers converge from below to the Fr\'echet distance, the Fr\'echet distance itself
		is uniformly lower semicomputable in $A$ and $B$.
	\end{proof}

It is easy to modify the proof to show computability of the Fr\'echet distance of surfaces which are parametrised over the
sphere $S^2$ rather than the square $[0,1]^2$.
The proof of Theorem \ref{Theorem: automorphisms located} can be used to show that the set of orientation-reversing automorphisms and the set of all automorphisms are computable as well.
This allows us to compute further variations of the Fr\'echet distance.

\section*{Acknowledgements}

	This work was supported by EU Horizon 2020 MSCA RISE project 731143.
	The majority of this work was undertaken while the author was visiting
	\textsc{KAIST}, Daejeon, Republic of Korea.
	The author would like to thank Martin Ziegler for bringing this problem to his attention.

\bibliographystyle{abbrv}
\bibliography{frechet}
\nocite{*}
\end{document}